\documentclass[showpacs,preprint,pra,superscriptaddress]{revtex4-1}

\usepackage[dvips]{graphicx} 
\usepackage{amsmath,amssymb,amsthm,amscd,mathrsfs,amsfonts,dsfont}
\usepackage{enumerate}
\usepackage{epsfig}
\usepackage{subfigure}
\usepackage{xcolor}
\usepackage{braket}
\usepackage{multirow}
\usepackage{subfigure}
\usepackage{dcolumn}
\usepackage{bm}
\usepackage{graphicx}

\newtheorem{theorem}{Theorem}
\newtheorem{lemma}{Lemma}

\begin{document}
\preprint{APS/123-QED}
\title{Quantum uncertainty relation using coherence}
\date{\today}
\author{Xiao Yuan}
\affiliation{Center for Quantum Information, Institute for Interdisciplinary Information Sciences, Tsinghua University, Beijing 100084, China}
\author{Ge Bai}
\affiliation{Department of Computer Science, The University of Hong Kong, Pokfulam Road, Hong Kong}
\author{Tianyi Peng}
\affiliation{Center for Quantum Information, Institute for Interdisciplinary Information Sciences, Tsinghua University, Beijing 100084, China}
\author{Xiongfeng Ma}
\email{xma@tsinghua.edu.cn}
\affiliation{Center for Quantum Information, Institute for Interdisciplinary Information Sciences, Tsinghua University, Beijing 100084, China}


\begin{abstract}
Measurement outcomes of a quantum state can be genuinely random (unpredictable) according to the basic laws of quantum mechanics. The Heisenberg-Robertson uncertainty relation puts constrains on the accuracy of two noncommuting observables. The existing uncertainty relations adopt variance or entropic measures, which are functions of observed outcome distributions, to quantify the uncertainty. According to recent studies of quantum coherence, such uncertainty measures contain both classical (predictable) and quantum (unpredictable) components. In order to extract out the quantum effects, we define quantum uncertainty to be the coherence of the state on the measurement basis. We discover a quantum uncertainty relation of coherence between two measurement non-commuting bases. Furthermore, we analytically derive the quantum uncertainty relation for the qubit case with three widely adopted coherence measures, the relative entropy using coherence, the coherence of formation, and the $l_1$ norm of coherence. 
\end{abstract}
\vspace{2pc}
\maketitle

\section{Introduction}
%
%
%
%
The uncertainty relation is one of the fundamental laws of quantum theory that differentiates it from the classical ones. Given two observables that have a nonzero average value of their commutator, the uncertainty relation constrains the possibility of measuring them simultaneously accurately. By adopting the variance as the measure for the uncertainty, the Heisenberg-Robertson uncertainty relation \cite{Robertson29} gives a lower bound of total variance of the two observables. Such uncertainty relations based on variance generally rely on how the observables are defined. Solely based on the observed statistics of the random measurement outcomes, Deutsch \cite{Deutsch83} and Maassen and Uffink \cite{Maassen88} further proposed uncertainty relations based on the Shannon entropy of the measurement outcomes. Considering two measurement bases $\mathbb{X} = \{\ket{x}\}$ and $\mathbb{Z} = \{\ket{z}\}$, the Maassen and Uffink uncertainty relation is given by ,
\begin{equation}\label{Eq:MUuncertainty}
  H(X) + H(Z) \ge - \log_2  c_{\mathrm{max}},
\end{equation}
where $X$ and $Z$ denote the measurement outcomes in the $\mathbb{X}$ and $\mathbb{Z}$ bases, respectively; $c_{\mathrm{max}}$ is an abbreviation of $c_{\mathrm{max}}(\mathbb{X},\mathbb{Z}) = \max_{\ket{x}\in \mathbb{X}, \ket{z}\in \mathbb{Z}}|\braket{x|z}|^2$ with a maximization over all measurement bases; and $H$ denotes the Shannon entropy of the measurement outcomes.

With the development of quantum information science, the uncertainty relation has been shown to be a powerful tool in the analyses of many quantum information tasks, including quantum key distribution \cite{Koashi2009Complement,berta2010uncertainty}, quantum random number generation \cite{Vallone14, Cao16}, entanglement witness \cite{Berta14}, EPR steering \cite{Walborn11, Schneeloch13}, and quantum metrology \cite{giovannetti2011advances}. With different measures of uncertainty, different forms of uncertainty relations have been derived. We refer to Ref.~\cite{coles2015entropic} for a recent review on this subject.

Essentially, an uncertainty relation characterizes the relationship between the nominal uncertainties of two measurements. Generally, the nominal uncertainty of a measurement contains classical (predictable) and quantum (unpredictable) parts that originate from classical noise and quantum effect, respectively \cite{YuanPhysRevA2015}. By saying predictable, we mean that the measurement outcome can be predicted by other systems.
For instance, considering the projective measurement in the $\mathbb{Z} = \{\ket{0},\ket{1}\}$ basis, the uncertainty of measuring $\rho=1/2(\ket{0}\bra{0}+\ket{1}\bra{1})$ originates from classical noise which is predictable; while the uncertainty of measuring $\ket{\psi}=1/\sqrt{2}(\ket{0}+\ket{1})$ originates from quantum effect which is unpredictable.
Alternatively, to see why the uncertainty of measuring $\rho=1/2(\ket{0}\bra{0}+\ket{1}\bra{1})$ is predictable, we consider an adversary Eve who holds its purification. Then, the joint state is $1/\sqrt{2}(\ket{00}+\ket{11})$ and the measurement result (in the $\mathbb{Z}$ basis) of the first system can always be predicted by the measurement result (in the $\mathbb{Z}$ basis) of the second system.

There is a fundamental difference between classical and quantum uncertainties as the unpredictability in measurement outcomes is a unique quantum feature.
Based on nominal uncertainties that involve both classical and quantum uncertainties, conventional uncertainty relations consist of both classical and quantum components.
To have a genuine quantum uncertainty relation, we should only use the quantum uncertainty.
Several efforts have been devoted along this line. By considering a joint system $\rho_{AB}$, Berta et al. \cite{berta2010uncertainty} proposed to use the conditional entropy to characterize the uncertainty of measuring system $A$ in the presence of the memory of system $B$. The state after $X$ basis measurement on system $A$ is given by $\rho_{XB}=\sum_x(\ket{x}\bra{x}_A\otimes I_B)\rho_{AB}(\ket{x}\bra{x}_A\otimes I_B)$, where $I_B$ denotes the identity matrix on system $B$.  The conditional entropy $H(X|B)$ of $\rho_{XB}$ thus measures the uncertainty of the measurement outcome conditioned on quantum system $B$. Define $H(Z|B)$ similarly, the uncertainty relation by Berta et al. is
\begin{equation}\label{}
  H(X|B) + H(Z|B) \ge -\log_2  c_{\mathrm{max}} + H(A|B).
\end{equation}
Here the conditional entropy $H(A|B)$ is given by $H(\rho_{AB})-H(\rho_{B})$, $\rho_{B}=Tr_A[\rho_{AB}]$, and $H$ represents the Von-Neumann entropy. With a slightly abuse of notation, we use $H$ to both represent the Shannon and Von Neumann entropy in this work.

The uncertainty relation proposed by Berta et al. depends on how system $B$ correlates with system $A$ before the measurement. To extract the genuine quantum uncertainty, we thus need to consider system $B$ as a purification of system $A$ and we can prove  \cite{yuan2016interplay} that,
\begin{equation}\label{Eq:quantumuncertainty}
\begin{aligned}
  C_{RE}^{\mathbb{X}}(\rho) \equiv H(X|B) &= H(X) - H(\rho_A), \\
  C_{RE}^{\mathbb{Z}}(\rho) \equiv H(Z|B) &= H(Z) - H(\rho_A).
\end{aligned}
\end{equation}
With the quantum uncertainties $C^\mathbb{X}_{\mathrm{RE}}(\rho)$ and $C^\mathbb{Z}_{\mathrm{RE}}(\rho)$, Korzekwa et al.  \cite{Korzekwa14} proposed another quantum uncertainty relation,
\begin{equation}\label{eq:Korzekwa}
  C^\mathbb{X}_{\mathrm{RE}}(\rho) + C^\mathbb{Z}_{\mathrm{RE}}(\rho) \ge -(1 -H(\rho))\log_2 c_{\mathrm{max}}.
\end{equation}

In general, the quantum uncertainty should only originate from quantum effects.
Recently, a coherence framework has been proposed to quantify the amount of quantumness in a given measurement basis \cite{Baumgratz14}. Furthermore, it is shown that the quantum (unpredictable) randomness or uncertainty measured on a basis corresponds to the coherence on the same basis \cite{YuanPhysRevA2015, yuan2016interplay}. Interestingly, the quantum uncertainty $C^\mathbb{X}_{\mathrm{RE}}(\rho)$ defined in Eq.~\eqref{Eq:quantumuncertainty} is also a coherence measure in the $X$ basis. As coherence is defined to quantify the quantumness on the measurement basis, it is reasonable to define coherence to be generalized quantum uncertainties. We refer to Ref.~\cite{streltsov2016quantum} for a detailed review of the coherence resource theory and its application in quantum information tasks.


In this paper, we investigate the quantum uncertainty relation using coherence on two measurement bases. As quantum uncertainty only originates from quantum effects, the lower bound of the uncertainty relation will also be genuinely quantum. Following this intuition, we first show that the uncertainty relation has a nontrivial lower bound as long as the state is not maximally mixed and the measurement bases satisfy a simple condition. Therefore, the quantum feature such as the purity of states is the main origin of the quantum uncertainty relation. Furthermore, we derive explicit quantum uncertainty relations for qubit states with several widely adopted coherence measures, the relative entropy of coherence, the $l_1$ norm of coherence, and the coherence of formation. We show that several conventional uncertainty relations can be slightly modified to be quantum uncertainty relations with the relative entropy of coherence. In comparison, we find that our result with the relative entropy of coherence outperforms existing ones when the purity of the state is not large and $c_{\mathrm{max}}$ not too small. In addition, we prove the tightness for the result with the $l_1$ norm of coherence.

\section{General quantum uncertainty---coherence}
In this section, we first review the resource framework for quantum coherence \cite{Baumgratz14}. We focus on a $d$-dimensional Hilbert space $\mathcal{H}_d$ and a computational basis $\mathbb{J} = \{\ket{1},\ket{2},\dots,\ket{d}\}$. A state $\sigma$ is called an \emph{incoherent} state when
\begin{equation}\label{Eq:incoherent}
  \sigma = \sum_{i}p_i\ket{i}\bra{i},
\end{equation}
where $p_i\in[0,1],\forall i\in\{1,2,\dots,d\}$ and $\sum_ip_i=1$. When a state $\rho$ cannot be written in the form of Eq.~\eqref{Eq:incoherent}, we call it a coherent state. Furthermore, incoherent operations are defined by physical operations that converts an incoherent state only to an incoherent state. With the definition of incoherent state and incoherent operation, the amount of coherence can be measured by a  function $C$ that maps state $\rho$ to a non-negative real value. In addition, coherence measure should satisfy the following properties.
\begin{enumerate}
  \item $C(\rho)=0$ for incoherent state and $C(\rho)>0$ for coherent state;
  \item $C(\rho)$ cannot increase under incoherent operations;
  \item $C(\rho)$ is convex.
\end{enumerate}
We refer to Ref.~\cite{Baumgratz14} for a detailed introduction of the resource framework of coherence.

Here, we focus on three widely adopted coherence measures, the relative entropy of coherence, the $l_1$ norm of coherence, and the coherence of formation.
The relative entropy \cite{Baumgratz14} of coherence is defined by
\begin{equation}\label{}
  C^\mathbb{J}_{\mathrm{RE}}(\rho) = H(\rho^{\mathrm{diag}}) - H(\rho),
\end{equation}
where $\rho^{\mathrm{diag}} = \sum_{i}\rho_{i,i}\ket{i}\bra{i}$ represents the state of $\rho$ after dephasing in the $\mathbb{J}$ basis, $\rho_{i,j}=\bra{i}\rho\ket{j}$ and $H$ is the von Neumann entropy of quantum states.
The coherence of formation \cite{aberg2006quantifying, YuanPhysRevA2015} is defined by
\begin{equation}\label{}
  C^\mathbb{J}_{\mathrm{CF}}(\rho) = \min_{p_e,\ket{\psi_e}}\sum_e p_e C^\mathbb{J}_{\mathrm{RE}}(\ket{\psi_e}\bra{\psi_e}),
\end{equation}
where the minimization is over all possible decompositions of $\rho = \sum_e p_e\ket{\psi_e}\bra{\psi_e}$, where $p_e$ is a probability distribution.
Finally, the $l_1$ norm of coherence \cite{Baumgratz14} is defined by
\begin{equation}\label{}
  C^\mathbb{J}_{l_1}(\rho) = \sum_{i\ne j}|\rho_{i,j}|,
\end{equation}
which is the sum of all the absolute values of the off-diagonal terms.

With the resource framework, we can see that coherence is indeed a good measure for quantum uncertainty. First, we consider the measurement outcome of an incoherent state $\sigma = \sum_{i}p_i\ket{i}\bra{i}$. It is straightforward to see that the measurement outcome (in $J$ basis) looks random. That is, the probability of obtaining the $i$th outcome is $p_i$. While, the measurement uncertainties of incoherent states originate from the classical noise of the state instead of genuine random features. We do not attribute such uncertainties as quantum uncertainties. Instead, it is shown that the measurement outcome contains genuine randomness (cannot be precisely predicted) as long as the state contains coherence on the measurement basis \cite{YuanPhysRevA2015,yuan2016interplay}. Therefore, when considering quantum uncertainty as the unpredictability of the measurement outcomes, then coherence supplies as a good measure.

In the following, we will investigate quantum uncertainty relation of coherence on two general measurement bases $\mathbb{X} = \{\ket{x}\}$ and $\mathbb{Z} = \{\ket{z}\}$. That is, we need to derive the following inequality,
\begin{equation}\label{Eq:uncertainty}
  C^\mathbb{X}(\rho) + C^\mathbb{Z}(\rho) \ge f(\mathbb{X},\mathbb{Z},\rho),
\end{equation}
where $C$ is a proper coherence measure and $f$ is a function of the measurement bases and the state $\rho$.


\section{Quantum uncertainty relation}
In this section, we present our result of quantum uncertainty relation. First, we show that there is always a nonzero lower bound to Eq.~\eqref{Eq:uncertainty} when the two bases satisfy a simple condition. Then, we explicitly propose the quantum uncertainty relation with the three measures and qubit states.

In the conventional uncertainty relation, such as the one in Eq.~\eqref{Eq:MUuncertainty}, the lower bound is state independent and only depends on the measurement bases. On the other hand, the uncertainty relation in Eq.~\eqref{Eq:uncertainty} will generally depend on the state $\rho$. For instance, when the state is maximally incoherent, i.e., $\rho = 1/d\sum_i\ket{i}\bra{i}$, which is an incoherent state for any basis, the lower bound is $f(\mathbb{X},\mathbb{Z},\rho)=0$ as we have $C^\mathbb{X}(\rho) = C^\mathbb{Z}(\rho) = 0$. Because the maximally incoherent state can be considered as a classical state for an arbitrary measurement basis, there will be no quantum uncertainty for the maximally incoherent state.

Now we consider the quantum uncertainty relation with states $\rho$ that is not maximally mixed. Intuitively, when the two measurement bases are not compatible, the quantum uncertainties in the two bases cannot be simultaneously zero. To be more rigourous, we define the incompatibility of two operators by $c_{\mathrm{min}} = \min_{\ket{x}\in \mathbb{X}, \ket{z}\in \mathbb{Z}}|\braket{x|z}|^2$ and prove the following result.

\begin{theorem}\label{Theorem:general}
The sum of the quantum uncertainties $C^\mathbb{X}(\rho) + C^\mathbb{Z}(\rho)$ in Eq.~\eqref{Eq:uncertainty} is positive for any state $\rho$ that is not maximally mixed and two measurement bases that are incompatible, i.e. $c_{\mathrm{min}}>0$.
\end{theorem}
\begin{proof}
Suppose the $C^\mathbb{X}(\rho) = C^\mathbb{Z}(\rho) = 0$. This means that $\rho$ is incoherent state in both the $\mathbb{X}$ and $\mathbb{Z}$ bases according to the properties of coherence measures. Therefore, we have that
\begin{equation}\label{}
  \rho = \sum_x p_x\ket{x}\bra{x} = \sum_z p_z\ket{z}\bra{z}.\nonumber
\end{equation}

Because $\rho\neq I/d$, $\rho$ has at least two different eigenvalues, say $\lambda_1, \lambda_2$. Denote the eigenvectors for $\lambda_i$ to be $\ket{x_i}$ and $\ket{z_i}$ in the $\mathbb{X}$ and $\mathbb{Z}$ bases, respectively. Then we have
\begin{equation}
\bra{x_1}\rho\ket{z_2} = \lambda_1\braket{x_1|z_2} = \lambda_2\braket{x_1|z_2}.\nonumber
\end{equation}
Because $|\braket{x_1|z_2}| \ge \sqrt{c_{\mathrm{min}}} > 0$, we thus have $\lambda_1 = \lambda_2$, which introduces the contradiction.
\end{proof}

In general, for a quantum state $\rho$ in a $d$-dimensional Hilbert space, we can make use of purity $P(\rho)=Tr[\rho^2]$ to measure how pure the state is. The purity value reaches its minimum $1/d$ for maximally mixed state and its maximum $1$ for pure state. In Theorem \ref{Theorem:general}, the condition that the state is not maximally mixed can thus also be expressed as $P(\rho)>1/d$. Furthermore, we explicitly show the quantum uncertainty relations for qubit states with the three different measures.

\subsection{Quantum uncertainty relation of the relative entropy of coherence}
The quantum uncertainty relation with the relative entropy of coherence can be derived from existing results. For example, one can apply the conventional uncertainty relation Eq.~\eqref{Eq:MUuncertainty} to obtain
\begin{equation}  \label{}
\begin{aligned}
  H(X) + H(Z) - 2 H(\rho) &\ge - \log_2 c - 2H(\rho),\nonumber
\end{aligned}
\end{equation}
that is
\begin{equation}\label{}
  C^\mathbb{X}_{\mathrm{RE}}(\rho) + C^\mathbb{Z}_{\mathrm{RE}}(\rho) \ge - \log_2 c - 2 H(\rho).
\end{equation}
Two strengthened versions of Eq.~\eqref{Eq:MUuncertainty} for general mixed states are respectively given by Berta et al. \cite{berta2010uncertainty} and Jorge Sanches-Ruiz \cite{Jorge98} as
\begin{equation}\label{Eq:berta1}
\begin{aligned}
  H(X) + H(Z) &\ge -\log_2 c_{\mathrm{max}} + H(\rho),\\
H(X) + H(Z) &\ge H\left(\frac{1+\sqrt{2c_{\mathrm{max}}-1}}{2}\right).
  \end{aligned}
\end{equation}
Thus, two tighter quantum uncertainty relations for the relative entropy of coherence are
\begin{equation}\label{eq:berta}
C^\mathbb{X}_{\mathrm{RE}}(\rho) + C^\mathbb{Z}_{\mathrm{RE}}(\rho) \ge - \log_2 c_{\mathrm{max}} - H(\rho),
\end{equation}
\begin{equation}\label{eq:sqrtc}
  C^\mathbb{X}_{\mathrm{RE}}(\rho) + C^\mathbb{Z}_{\mathrm{RE}}(\rho) \ge H\left(\frac{1+\sqrt{2c_{\mathrm{max}}-1}}{2}\right) - 2H(\rho).
\end{equation}
A similar relation to Eq.~\eqref{eq:berta} is derived by Singh et al.~\cite{math4030047}. Another quantum uncertainty relation is given by Korzekwa et al.~\cite{Korzekwa14} as in Eq.~\eqref{eq:Korzekwa}. In our work, by considering the geometric structure of qubit states, we also propose a different quantum uncertainty relation with the relative entropy of coherence. For simplicity, we denote $c_{\mathrm{max}}$ as $c$ and $P$ to be the purity of the state $\mathrm{Tr}[\rho^2]$. Note that we also have $c = 1-c_{\mathrm{min}}$.

\begin{theorem}\label{Theo:relative}
Given a qubit state $\rho$ and two measurement bases $\mathbb{X} = \{\ket{x}\}$ and $\mathbb{Z} = \{\ket{z}\}$, the quantum uncertainty relation of the relative entropy of coherence is
\begin{equation}\label{eq:result2}
  C^\mathbb{X}_{\mathrm{RE}}(\rho) + C^\mathbb{Z}_{\mathrm{RE}}(\rho) \ge H\left(\frac{\sqrt{2P-1}(2\sqrt{c}-1)+1}{2}\right) - H\left(\rho\right).
\end{equation}
\end{theorem}
Note that $C^{\mathbb{X}}_{RE}(\rho) = H(X) - H(\rho)$, thus the quantum uncertainty relation in Eq.  \eqref{eq:result2} is equivalent to a conventional uncertainty relation
\begin{equation}
  H(X) + H(Z) \geq  H\left(\frac{\sqrt{2P-1}(2\sqrt{c}-1)+1}{2}\right)+H(\rho).
\end{equation}
For a qubit state $\rho$, we consider its spectral decomposition as $\rho=p\ket{r}\bra{r}+(1-p)\ket{r_{\perp}}\bra{r_{\perp}}$ and the entropy $H(X)$ and $H(Z)$ can be rewritten as:
\begin{equation}
\begin{aligned}
  H(X) &= H(ap+(1-a)(1-p)),\\
  H(Z) &= H(bp+(1-b)(1-p)),
\end{aligned}
\end{equation}
where $a=|\braket{r|x}|^2$ and $b=|\braket{r|z}|^2$, $\ket{x}, \ket{z}$ is an arbitrary base vector in $\mathbb{X},\mathbb{Z}$, respectively. Note that the purity of the state $P$ is a function of $p$, i.e., $P = 2p^2-2p+1$.

Our aim is to find the relation between $H(X)+H(Z)$ and $c= \max_{\ket{x}\in \mathbb{X}, \ket{z}\in \mathbb{Z}}|\braket{x|z}|^2$. Without loss of generality, we can assume that $c = |\braket{x|z}|^2$. Otherwise, we can always choose $\ket{x}$ and $\ket{z}$ such that their inner product is maximized. Note that, $a$, $b$, and $c$ are the square of the inner product between each two vectors of the three normalized vectors $\ket{x},\ket{r},\ket{z}$. Intuitively,  when $c$ is large, $\ket{x}$ and $\ket{z}$ should be close and the difference between $a$ and $b$ also should be small. When c is small, i.e., $\ket{x}$ and $\ket{z}$ become more orthogonal, the sum of $a$ and $b$ should be near to one. This intuition is summarized as the following Lemma,


\begin{lemma}\label{Lemma:reLem}
For any three normalized vectors $\ket{x}$, $\ket{z}$ and $\ket{r} \in \mathbb{C}^{d}$, define $a=|\braket{r|x}|^2,b=|\braket{r|z}|^2,c=|\braket{x|z}|^2$, then
\begin{equation}
  a+b \leq 1+\sqrt{c}, |a-b|\leq \sqrt{1-c}.
\end{equation}
When $d=2$, we also have that
\begin{equation}
  1-\sqrt{c} \leq a+b.
\end{equation}
\end{lemma}

\begin{proof}
Firstly we prove $a+b\leq 1+\sqrt{c}$. Given a suitable basis, we can assume that $\ket{x}=(1,0,...,0)$. Because the values of $a$, $b$ and $c$ are invariant by adding a constant phase $e^{i\theta}$ to $\ket{z}$ or $\ket{r}$, we can thus assume the first dimension of $\ket{z}$ and $\ket{r}$ to be non-negative real numbers. Denote $\ket{z}=(\cos\frac{\alpha}{2}, \sin\frac{\alpha}{2}\ket{w})$, $\ket{r} = (\cos \frac{\theta}{2}, \sin \frac{\theta}{2} \ket{s})$, where $\alpha,\theta\in [0,\pi]$ and $\ket{w}, \ket{s}$ are normalized vectors in $\mathbb{C}^{d-1}$. Then

\begin{equation} \label{eq:sumAB}
  \begin{aligned}
  &a + b \\
  =& |\braket{x|r}|^2+|\braket{z|r}|^2\\
  =& \cos^2 \frac{\theta}{2} + \left|\cos \frac{\theta}{2}\cos \frac{\alpha}{2} + \sin \frac{\theta}{2}\sin \frac{\alpha}{2}\braket{w|s}\right|^2\\
  \leq &\cos^2 \frac{\theta}{2} + \left(\left|\cos \frac{\theta}{2} \cos \frac{\alpha}{2}\right| + \left|\sin\frac{\theta}{2}\sin \frac{\alpha}{2}\braket{w|s} \right|\right)^2\\
  \leq& \cos^2 \frac{\theta}{2} + \left(\left|\cos \frac{\theta}{2} \cos \frac{\alpha}{2}\right| + \left|\sin\frac{\theta}{2}\sin \frac{\alpha}{2}\right| \right)^2 ~\alpha,\theta \in [0,\pi]\\
  =&\cos^2 \frac{\theta}{2} + \cos^2 (\frac{\alpha}{2} - \frac{\theta}{2}) \\
  =&1 + \frac{1}{2}(\cos \theta + \cos (\alpha - \theta))\\
  =&1 + \cos \frac{\alpha}{2} \cos \left(\theta - \frac{\alpha}{2} \right)\\
  \leq& 1 + \cos \frac{\alpha}{2}\\
  =& 1+ \sqrt{c}.
  \end{aligned}
\end{equation}

Next, we prove $|a-b|\leq \sqrt{1-c}$. We denote
\begin{eqnarray}
   \ket{z_{\perp}}_{r} &=& \frac{\ket{r}-\braket{z|r}\ket{z}}{\left| \ket{r}-\braket{z|r}\ket{z} \right|} = \frac{\ket{r} - \braket{z|r}\ket{z}}{\sqrt{1-b}},  \nonumber\\
   \ket{z_{\perp}}_{x} &=& \frac{\ket{x}-\braket{z|x}\ket{z}}{\left| \ket{x}-\braket{z|x}\ket{z} \right|} = \frac{\ket{x} - \braket{z|x}\ket{z}}{\sqrt{1-c}}\nonumber.
\end{eqnarray}
Then, applying the first result $|\braket{x|r}|^2+|\braket{z|r}|^2 \leq 1+\sqrt{|\braket{x|z}|^2}$ by replacing $\ket{z}$ with $\ket{z_{\perp}}_r$, we have
\begin{equation}
  \begin{aligned}
  a + |\braket{r|z_{\perp}}_r|^2 &\leq 1 + \sqrt{|\braket{x|z_{\perp}}_r|^2}\\
  a + 1 - b &\leq 1 + \sqrt{|(\braket{x|z}\bra{z}+\braket{x|z_{\perp}}_x\bra{z_\perp}_x)\ket{z_\perp}_r|^2}\\
  a - b &\leq \sqrt{|\braket{x|z_{\perp}}_x\bra{z_\perp}_x\ket{z_\perp}_r|^2}\\
        &= \sqrt{|\braket{x|z_{\perp}}_x|^2|\bra{z_\perp}_x\ket{z_\perp}_r|^2}\\
        &\leq \sqrt{|\braket{x|z_{\perp}}_x|^2}\\
        &= \sqrt{1-c}.
  \end{aligned}
\end{equation}
Similarly, we have $b-a \leq \sqrt{1-c}$ and hence $|a-b|\leq \sqrt{1-c}$.

When considering in two dimensions, we can apply the first result $|\braket{x|r}|^2+|\braket{z|r}|^2 \leq 1+\sqrt{|\braket{x|z}|^2}$ by replacing $\ket{x}=(x_1, y_1)$ with $\ket{x_\perp}=(-y_1,x_1)$, $\ket{z}=(x_2,y_2)$ with $\ket{z_\perp}=(-y_2,x_2)$. Then, we can check that $|\braket{r|x_\perp}|^2 = 1 - a$, $|\braket{r|z_\perp}|^2 = 1 - b$, and $|\braket{x_\perp|z_\perp}|^2 = c$, and we have
\begin{equation}
  \begin{aligned}
    |\braket{r|x_\perp}|^2 + |\braket{r|z_\perp}|^2 &\leq 1 + \sqrt{|\braket{x_\perp|z_\perp}|^2}\\
    1 - a + 1 - b &\leq 1 + \sqrt{c}\\
    a + b &\ge 1 - \sqrt{c}\\
  \end{aligned}
\end{equation}
\end{proof}

Now, we prove Theorem~\ref{Theo:relative}.
\begin{proof}
Denote $f(x)=H(xp+(1-x)(1-p))$, then we need to find the minimal value of $g = f(a) + f(b)$. Because $H$ is a concave function and symmetrical about $x=\frac{1}{2}$, it is easy to check that $f$ is also a concave function which is symmetrical about $x=\frac{1}{2}$:
\begin{equation}
\begin{aligned}
\frac{\partial^2 f}{\partial x^2} &= H''(xp+(1-x)(1-p))(2p-1)^2 \leq 0.\\
f(x) &= H(xp+(1-x)(1-p))\\
 &=H(1-xp-(1-x)(1-p))\\
&= H((1-x)p+x(1-p))\\
&= f(1-x).
\end{aligned}
\end{equation}
So the maximal value of $f$ is $1$ with $x={1}/{2}$. The minimal value of $f$ is $H(p)$ with $x=0$ or $x=1$.

Denote $A = a + b$ and $B = b - a$, then we have $g = f\left(\frac{A+B}{2}\right)+f\left(\frac{A-B}{2}\right)$. By Lemma~\ref{Lemma:reLem},  there is $1-\sqrt{c}\leq a+b \leq 1+\sqrt{c}$ and $b-a\leq \sqrt{1-c}$. Hence $A\in [1-\sqrt{c},1+\sqrt{c}]$ and $B\in [0, \sqrt{1-c}]$. Here, without loss of generality, we assume $a\le b$. Furthermore, as $a,b\in[0,1]$, we have that $A+B\in[0,2]$ and $A-B\in[0,2]$.
The first and second partial derivatives of $g$ over $A$ and $B$ are
\begin{equation}\label{}
\begin{aligned}
\frac{\partial g}{\partial A} &= \frac{f'(b) + f'(a)}{2},\\
\frac{\partial g}{\partial B} &= \frac{f'(b) - f'(a)}{2} \leq 0,\\
\frac{\partial^2 g}{\partial A^2}  = \frac{\partial^2 g}{\partial B^2} &= \frac{f''(a) + f''(b)}{4} \leq 0.\\
\end{aligned}
\end{equation}
In addition, $g$ is symmetrical about $A = 1$, i.e.,
\begin{align}\label{}
\begin{aligned}
   &f\left(\frac{A+B}{2}\right) + f\left(\frac{A-B}{2}\right) \\
   &= f\left(1-\frac{A+B}{2}\right)+f\left(1-\frac{A-B}{2}\right).
\end{aligned}
\end{align}
In this case, we have $\frac{\partial g}{\partial A}\ge0$ for $A\le1$ and $\frac{\partial g}{\partial A}\le0$ for $A>1$.

As $g$ is symmetric about $A = 1$, we only consider that $A\le 1$. To find the minimal value of $g$, we need to choose the value of $B$ as large as possible and the value of $A$ as small as possible. Considering the additional constraints $A+B\in[0,2]$ and $A-B\in[0,2]$, the minimal value must be obtained in one of following points as shown in Fig.~\ref{fig:FigRange},
\begin{itemize}
\item $A=B=1-\sqrt{c}$, and $g = f(0) + f(\sqrt{c}).$
\item $A=B=\sqrt{1-c}$, and  $g= f(0) + f(\sqrt{1-c}).$
\end{itemize}
Because $c\geq {1}/{2}$, we can verify that
\begin{equation}
\left|\sqrt{c}-\frac{1}{2}\right| \ge \left|\sqrt{1-c} - \frac{1}{2} \right|.
\end{equation}
Hence because $f$ is concave and symmetric about $1/2$ we conclude that $f(\sqrt{c}) \leq f(\sqrt{1-c})$.

\begin{figure}[hbt]
\centering \resizebox{8cm}{!}{\includegraphics{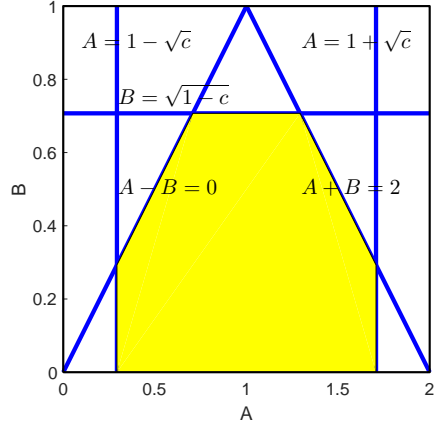}}
\caption{Blues solid lines are used to limit the range of A and B. Yellow area represents the feasible area. In this figure, $c=0.5$.}
\label{fig:FigRange}
\end{figure}

As a result,  the minimal value of $g$ is given by $f(0)+f(\sqrt{c}) = H(p) + H(p\sqrt{c}+(1-p)\sqrt{c})$ and the quantum uncertainty relation of the relative entropy of coherence is
\begin{equation}
\begin{aligned}
& C^\mathbb{X}_{\mathrm{RE}}(\rho) + C^\mathbb{Z}_{\mathrm{RE}}(\rho)\\
&=H(X) + H(Z) - 2*H(\rho)\\
&= f(a) + f(b) - 2*H(\rho)\\
&\geq f(0) + f(\sqrt{c}) - 2*H(\rho)\\
&= H\left(\frac{\sqrt{2P-1}(2\sqrt{c}-1)+1}{2}\right) - H(\rho).
\end{aligned}
\end{equation}
\end{proof}
\subsection{Quantum uncertainty relation of the coherence of formation}\label{app:Fm}
Now, we consider the quantum uncertainty relation of the coherence of formation.
\begin{theorem}
For qubit state $\rho$ and two measurement bases $\mathbb{X} = \{\ket{x}\}$ and $\mathbb{Z} = \{\ket{z}\}$, we have
\begin{equation}
  C^\mathbb{X}_{\mathrm{CF}}(\rho) + C^\mathbb{Z}_{\mathrm{CF}}(\rho) \ge H\left(\frac{1+\sqrt{1-4(2P-1)\sqrt{c}(1-\sqrt{c})}}{2}\right).
\end{equation}
\end{theorem}

\begin{proof}
Given the spectral decomposition of $\rho= p\ket{r}\bra{r} + (1-p)\ket{r_\perp}\bra{r_\perp}$ we have $\sqrt{2P-1} = |2p-1|$. Suppose $c$ is given by $|\braket{x|z}|^2$, where $\ket{x},\ket{z}$ is a base vector in $\mathbb{X}$, $\mathbb{Z}$, respectively. The coherence of formation $C^{\mathbb{X}}_{CF}(\rho)$ can be rewritten as
\begin{equation}
  \begin{aligned}
  &C^{\mathbb{X}}_{CF}(\rho)\\
  &=H\left(\frac{1+\sqrt{1-C_{l_1}^{\mathbb{X}}(\rho)}}{2}\right),\\
  &=H\left(\frac{1+\sqrt{1-2\sqrt{2P-1}\sqrt{a}\sqrt{1-a}}}{2}\right).
  \end{aligned}
\end{equation}
Similarly,
\begin{equation}
  C^{\mathbb{Z}}_{CF}(\rho) = H\left(\frac{1+\sqrt{1-2\sqrt{2P-1}\sqrt{b}\sqrt{1-b}}}{2}\right).
\end{equation}
Denote function $g(x) = \frac{1+\sqrt{1-2\sqrt{2P-1}\sqrt{x}\sqrt{1-x}}}{2}, f(x) = H\left(g(x)\right)$, then
\begin{equation}
  C^{\mathbb{X}}_{CF}(\rho) + C^{\mathbb{Z}}_{CF}(\rho) = f(a) + f(b).
\end{equation}
Firstly, it is easy to verify that $g(x)$ is a convex function. In addition, we have $g(x) \geq \frac{1}{2}, H'(x)\leq 0$ when $x\ge {1}/{2}$. As $H$ is a concave function, we prove that $f$ is a concave function
\begin{equation}
  \begin{aligned}
  f''(x) = H''(g(x))g'(x)^2 + g''(x)H'(g(x)) \leq 0.
  \end{aligned}
\end{equation}
Also, $f$ is symmetrical about $x=\frac{1}{2}$:
\begin{equation}
  f(x) = f(1-x).
\end{equation}
Therefore, for fixed $c$ and $P$, with a similar the proof of Theorem~\ref{Theo:relative}, we can conclude,
\begin{equation}
\begin{aligned}
  &C^{\mathbb{X}}_{CF}(\rho) + C^{\mathbb{Z}}_{CF}(\rho),\\
  &= f(a) + f(b), \\
  &\geq f(0) + f(\sqrt{c}),\\
  &= H\left(\frac{1+\sqrt{1-2\sqrt{2P-1}\sqrt{c}\sqrt{1-c}}}{2}\right).
\end{aligned}
\end{equation}
\end{proof}

\subsection{Quantum uncertainty relation of the $l_1$ norm of coherence}\label{app:L1}
Now we derive the quantum uncertainty relation of the $l_1$ norm of coherence.
For a qubit state $\rho$ and a measurement basis $\mathbb{X}=\{x,x_\perp\}$, the $l_1$ norm of coherence $C^{\mathbb{X}}_{l_1}$ is
\begin{equation}
C_{l_1}^{\mathbb{X}} = 2|\bra{x}\rho \ket{x_\perp}|.
\end{equation}

To prove the lower bound of $C_{l_1}^{\mathbb{X}}$, we firstly introduce a lemma of three dimensional space:
\begin{lemma}\label{ThreeDlm}
Suppose $\vec{a},\vec{b},\vec{c}\in \mathbb{R}^3$ are three dimensional nonzero vectors. Denote $\alpha,\beta,\gamma$ to be the angle between $\vec{a},\vec{b}$; $\vec{b},\vec{c}$ and $\vec{c},\vec{a}$, respectively ($\alpha,\beta,\gamma \in [0,\pi]$). There is
\begin{equation}
\sin \alpha + \sin \beta \ge \sin \gamma
\end{equation}
\end{lemma}
\begin{proof}
When $\vec{a},\vec{b},\vec{c}$ are in the same plane, we have $\gamma = 2\pi - \alpha - \beta$ and hence
\begin{equation}
\begin{aligned}
&\sin \alpha + \sin \beta - \sin \gamma \\
=& \sin \alpha + \sin \beta - \sin (2\pi - \alpha - \beta)\\
=&\sin \alpha (1 + \cos\beta) + \sin \beta (1+  \cos\alpha)\\
\ge&0. \\
\end{aligned}
\end{equation}
The last inequality comes as $\alpha,\beta \in [0,\pi]$ and hence $\sin \alpha, (1 + \cos\beta), \sin \beta, (1+  \cos\alpha)\ge0$.

When $\vec{a},\vec{b},\vec{c}$ are not in the same plane, we can consider the plane formed by $\vec{a},\vec{c}$ and project $\vec{b}$ on this plane as $\vec{b'}$. Without loss of generality, assume $\vec{a}=(1,0,0),\vec{b}=(b_x,b_y,b_z),\vec{c}=(c_x,c_y,0)$ such that $b_x^2+b_y^2+b_z^2 = c_x^2+c_y^2=1$. Then we have
\begin{equation}
\begin{aligned}
&b_x^2(1-b_x^2-b_y^2) \ge 0\\
\implies& (b_x^2+b_y^2)(1-b_x^2) \ge b_y^2\\
\implies& b_y^2+b_z^2 \ge \frac{b_y^2}{b_x^2+b_y^2}\\
\implies& \sin^2\alpha \ge \sin^2 \alpha' \\
\implies& \sin\alpha \ge \sin \alpha'
\end{aligned}
\end{equation}
By symmetric, $\sin\beta \ge \sin\beta'$, where $\beta'$ is the angle between $\vec{b'}$ and $\vec{c}$. In general, we have
\begin{equation}
\sin\alpha+\sin\beta \geq \sin\alpha'+\sin\beta' \geq \sin \gamma.
\end{equation}
\end{proof}

Now consider a qubit state and its Bloch sphere representation
\begin{equation}
  \rho = \frac{I+\vec{r} \cdot \vec{\sigma}}{2},
\end{equation}
where $\vec{\sigma} = (\sigma_x, \sigma_y, \sigma_z)$ are the Pauli matrices. Then we have the following result.
\begin{lemma}\label{BSLe}
For two pure qubit states $\ket{x},\ket{z}$ and their corresponding vectors in the Bloch sphere as $\vec{x}',\vec{z}'$. Suppose the angle between $x'$ and $z'$ are $\alpha$, then we have $|\braket{x|z}|^2 = \cos^2\frac{\alpha}{2}$.
\end{lemma}
\begin{proof}

Therefore
\begin{equation}
\begin{aligned}
  &|\braket{x|z}|^2\\
  &= Tr(\ket{x}\bra{x}\ket{z}\bra{z}) \\
  &= Tr\left(\frac{I+\vec{x}'\vec{\sigma}}{2} \frac{I+\vec{z}'\vec{\sigma}}{2}\right) \\
  &= \frac{1}{2}(1 + \vec{x}'\cdot \vec{z}')\\
  &= \frac{1}{2}(1+\cos\alpha)\\
  &= \cos^2\frac{\alpha}{2}.
\end{aligned}
\end{equation}
\end{proof}

With Lemma~\ref{ThreeDlm} and \ref{BSLe}, we can prove the quantum uncertainty relation of the $l_1$ norm of coherence,
\begin{theorem}
For qubit state $\rho$ and two measurement bases $\mathbb{X} = \{\ket{x}\}$ and $\mathbb{Z} = \{\ket{z}\}$, we have
\begin{equation}
  C^\mathbb{X}_{l_1}(\rho) + C^\mathbb{Z}_{l_1}(\rho) \ge 2\sqrt{(2P-1)c(1-c)}.
\end{equation}
\end{theorem}

\begin{proof}
Suppose the spectral decomposition of $\rho$ is $\rho=p\ket{r}\bra{r} + (1-p)\ket{r_\perp}\bra{r_\perp}$ and we have $(2p-1)^2 = 2P-1$.
Suppose $c$ is given by $|\braket{x|z}|^2$, where $\ket{x},\ket{z}$ is a base vector in $\mathbb{X},\mathbb{Z}$, respectively.

Denote $\alpha,\beta,\gamma$ to be the angles in the Bloch sphere between the corresponding vectors of $\ket{x},\ket{z}$; $\ket{z},\ket{r}$ and $\ket{x},\ket{z}$ respectively. By Lemma~ \ref{BSLe}, we have $a=|\braket{x|r}|^2 = \cos^2\frac{\alpha}{2}$, $b=|\braket{z|r}|^2=\cos^2\frac{\beta}{2}$, and $c=\cos^2\frac{\gamma}{2}$. The $l_1$ norm of coherence of $\rho$ in the $\mathbb{X}$ basis is
\begin{equation}
  \begin{aligned}
  &C_{l_1}^\mathbb{X}(\rho)\\
  &=2|\bra{x}\rho\ket{x_\perp}|\\
  &=2|p\braket{x|r}\braket{r|x_\perp}+(1-p)\braket{x|r_\perp}\braket{r_\perp|x_\perp}|\\
  &=2|p\braket{x|r}\braket{r|x_\perp}+(1-p)\bra{x}(I-\ket{r}\bra{r})\ket{x_\perp}|\\
  &=2|p\braket{x|r}\braket{r|x_\perp}-(1-p)\braket{x|r}\braket{r|x_\perp}|\\
  &=2|2p-1||\braket{r|x}||\braket{r|x_\perp}|\\
  &=2\sqrt{2P-1}\sqrt{a}\sqrt{1-a}\\
  &=2\sqrt{2P-1}\cos\frac{\alpha}{2}\sin\frac{\alpha}{2}\\
  &=\sqrt{2P-1}\sin\alpha.
  \end{aligned}
\end{equation}

Similarly, we have $C^\mathbb{Z}_{l_1} = \sqrt{2P-1}\sin\beta$. By Lemma~\ref{ThreeDlm}, we have
\begin{equation}
  \begin{aligned}
    &C_{l_1}^{\mathbb{X}} + C_{l_1}^{\mathbb{Z}}\\
    &=\sqrt{2P-1}(\sin\alpha + \sin\beta)\\
    &\geq \sqrt{2P-1}\sin \gamma\\
    &=2\sqrt{2P-1}\sqrt{c(1-c)}.
  \end{aligned}
\end{equation}

It is easy to construct three vectors in Bloch sphere such that $\alpha=0, \beta = \gamma$ and $\sin\alpha+\sin\beta=\sin\gamma$. Therefore, the bound $2\sqrt{2P-1}\sqrt{c(1-c)}$ can be saturated for $C_{l_1}^{\mathbb{X}}+C_{l_2}^{\mathbb{Z}}$ with fixed $c$ and $P$.
\end{proof}

\section{Numerical simulation and comparison}
In this section, we numerically analyze our results. First, we propose a numerical method to calculate the tight lower bound of coherence measures. Then, we compare our results with existing ones.

\subsection{Optimal Numerical Bound for Three Measures}\label{numerical}
The quantum uncertainty relation of the $l_1$ norm of coherence is tight. That is, there always exists a quantum state that saturates the equal sign. However, due to the complexity of the Shannon entropy function, the quantum uncertainty relations of the other two measures are not tight. An efficient numerical approach for the conventional uncertainty relation with qubit states has been proposed in Ref.~\cite{Jorge98}. In this section, we generalize the result and propose a numerical method to calculate the tight lower bound to general quantum uncertainty relation with coherence measures that are concave and symmetric functions.

The three coherence measures considered in this work can be transformed to the following problem:
\begin{equation}\label{Eq:minimize}
\begin{aligned}
&\min g = f(a)+f(b)\\
s.t.~~ & |\braket{x|z}|^2=c,~ c\in [1/2, 1]\\
& |\braket{x|r}|^2 = a\\
& |\braket{z|r}|^2 = b
\end{aligned}
\end{equation}
Specifically, for relative entropy measure: $f(x)=H(xp+(1-x)(1-p))-H(p)$; for $l_1$ measure, $f(x)=2|2p-1|\sqrt{x(1-x)}$; for coherence of formation, $f(x) = H\left(\frac{1+\sqrt{1-2|2p-1|\sqrt{x(1-x)}}}{2}\right)$.
Denote the Bloch sphere representations of the state vectors $\ket{x}$, $\ket{z}$, and $\ket{r}$ as $\vec{x}$, $\vec{z}$, and $\vec{r}$, respectively. Suppose the angles between $\vec{x}$ and $\vec{z}$, $\vec{z}$ and $\vec{r}$, $\vec{r}$ and $\vec{x}$ are $\alpha$, $\beta$, $\gamma$. Then, according to Lemma~\ref{BSLe}, the constraints becomes $a = cos^{2}\frac{\alpha}{2}=(\cos \alpha + 1)/2$, $b = (\cos \beta + 1)/2$, $c = (\cos \gamma + 1)/2$ and the minimization problem in Eq.~\ref{Eq:minimize} becomes

\begin{equation}
\begin{aligned}
&\min f((\cos(\alpha)+1)/2)+f((\cos(\beta)+1)/2)\\
s.t.~~ & \gamma \leq \alpha+\beta \leq 2\pi - \gamma\\
& 0\leq \alpha - \beta \leq \gamma\\
&\gamma \in [0,\pi/2]
\end{aligned}
\end{equation}

When the the function $f(x)$ is concave and symmetrical about $x=\frac{1}{2}$, we can further simply the minimization with the following lemma.
\begin{lemma}\label{angle}
Suppose $f$ is concave and symmetrical about $x=\frac{1}{2}$. Denote $g(A,B) = f(\frac{\cos\alpha+1}{2})+f(\frac{\cos\beta+1}{2})$ where $\alpha = \frac{A+B}{2}$, $\beta = \frac{A-B}{2}$, $\alpha,\beta \in [0,\pi]$. Then $g(A,B)$ is symmetrical about $A=\pi$ and concave about $A$ with fixed $B$.
\end{lemma}

\begin{proof}
Denote $h(\alpha) = \frac{\cos\alpha+1}{2}$, then we have $f(h(\alpha)) = f(h(\pi-\alpha))$ as $h(\alpha)+h(\pi-\alpha)=1$. So $f(h(\alpha))$ is symmetrical about $\alpha = \frac{\pi}{2}$. Also, we have
\begin{equation}
\frac{\partial^2 f(h(\alpha))}{\partial \alpha^2} = f''(h(\alpha))h'(\alpha)^2 + h''(\alpha)f'(h(\alpha))
\end{equation}
where $f''(h(\alpha))\leq 0$. It is easy to see that when $\alpha \in [0,\frac{\pi}{2}]$, we have $h''(\alpha) = \frac{-\cos \alpha}{2} \leq 0$ and  $f'(h(\alpha)) \geq 0$; when $\alpha \in [\frac{\pi}{2},\pi]$, we have $h''(\alpha) \geq 0$ and $f'(h(\alpha)) \leq 0$. Therefore, $f(h(\alpha))$ is a concave function.

Now for fixed $B$, we have
\begin{equation}
\frac{\partial^2 g}{\partial A^2} = \left(\frac{\partial^2 f(h(\alpha))}{\partial \alpha^2} +\frac{\partial^2 f(h(\beta))}{\partial \beta^2}\right)/4 \leq 0.
\end{equation}
Also $g(A,B)=f(h((A+B)/2))+f(h((A-B)/2))=f(h(\pi-(A+B)/2))+f(h(\pi-(A-B)/2)) = g(2\pi-A, B)$ implies that $g(A,B)$ is symmetrical about $A=\pi$.

\end{proof}

From Lemma~\ref{angle}, we can derive that for fixed $\alpha-\beta$, the minimal value is obtained with $\alpha+\beta = \gamma$. So the problem can be further simplified  to a single variable optimization problem:
\begin{equation}
\begin{aligned}
&\min f((\cos(\alpha)+1)/2)+f((\cos(\gamma-\alpha)+1)/2)\\
s.t.~~ & \gamma/2\leq \alpha \leq \gamma
\end{aligned}
\end{equation}
Which can be solved by a numerical search.


\subsection{Comparison with existing results}
In comparison, we plot in Figure \ref{fig:compare} the lower bounds of the four results in Eq.~\eqref{eq:Korzekwa}  (from Ref.~\cite{Korzekwa14}), Eq.~\eqref{eq:berta} (from Ref.~\cite{berta2010uncertainty}), Eq.~\eqref{eq:sqrtc} (from Ref.~\cite{Jorge98}), and Eq.~\eqref{eq:result2} (our result) with different bases $c$ and purity $P$ of the states. In addition, we also plot the numerical tight bound with a method described in Section \ref{numerical}.
In comparison, our result is less optimal than Eq.~\eqref{eq:sqrtc} when $\rho$ has a larger purity; however our result gives a much stricter bound when the purity of the state is low. Compared to Eq. \eqref{eq:Korzekwa}, our bound is better when $c$ is larger than a certain value.
In summary, we can see that our result outperforms the existing ones when the purity is not large and $c$ not too small.

\begin{figure*}[hbt]\centering
\resizebox{16cm}{!}{\includegraphics{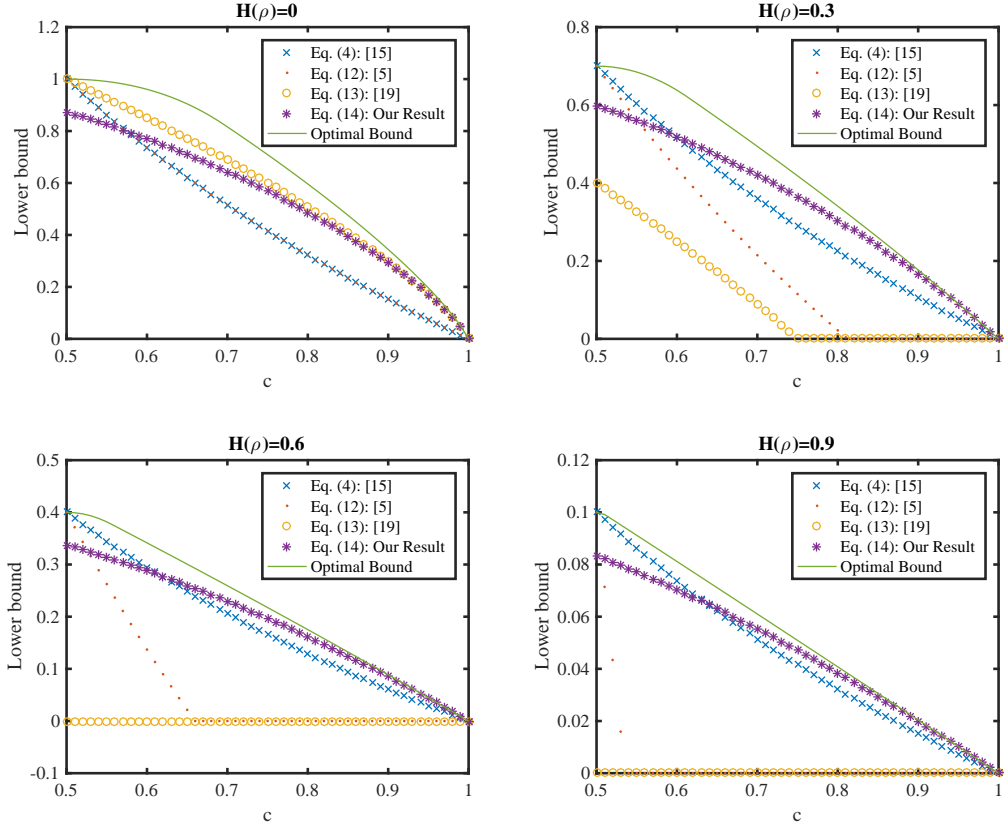}}
\caption{The various results of uncertainty relation using coherence under relative entropy measure. The blue line is our analytic bound and the green line is our optimal numerical bound. }
\label{fig:compare}
\end{figure*}



\section{Conclusion}
In this paper, we introduce a concept of quantum uncertainty relations using coherence. Compared to conventional relations, we only consider the uncertainty that is introduced from quantum effects. Hence, the uncertainty relation connects the true randomness introduced by quantum effects on two bases. The lower bound indicates the amount of quantumness the state possesses. For the qubit case, we derive analytical relations for the three widely adopted coherence measures, relative entropy of coherence, $l_1$ norm of coherence, and coherence of formation. Quantum coherence also plays important roles in quantum optics \cite{PhysRevA.93.032111, PhysRevA.93.032111}, thus applying the quantum uncertainty relation in quantum optics is an interesting subject for future work.

Generalizations of the results to general qudit states are natural extensions of this work. Here, we only consider the largest inner product $c_{\mathrm{max}}$ of the two bases. When considering general qudit states, the overlap $c(x,z) = |\braket{z|x}|^2$ of the two bases cannot be simply characterized by the largest inner product $c_{\mathrm{max}}$. In Ref.~\cite{Coles14}, the authors proposed to additionally make use of the second largest value of $c(x,z)$ to measure the lower bound. For deriving quantum uncertainty for general qudit states, such technique would be useful for deriving a tighter bound.

\section*{Acknowledgements}
This work was supported by the National Natural Science Foundation of China Grant No.~11674193.

\bibliographystyle{apsrev4-1}
\bibliography{bibCU}




\end{document}